\documentclass[10pt,a4paper]{article}
\usepackage{amsmath,amsthm,amssymb,pstricks,pst-node}
\usepackage{graphics,xcolor}
\usepackage{geometry}

\usepackage{amsmath,amsthm,amssymb,pstricks,pst-node,paralist,graphics,epsfig}
\usepackage{float,enumerate,fancyhdr}

\parskip=\medskipamount
 1   
 1  

\newtheorem{theorem}{Theorem}[section]

\newtheorem{lemma}[theorem]{Lemma}
\newtheorem{proposition}{Proposition}

\theoremstyle{definition}
\newtheorem{definition}[theorem]{Definition}

\newcommand{\A}{{\mathcal A}}

\newcommand{\Ac}{{\mathfrak{A}}}
\newcommand{\B}{{\mathcal B}}
\newcommand{\Ro}{{\mathcal R}}

\newcommand{\e}{\epsilon}
\newcommand{\lag}{\mathfrak{g}}

\def\fpd#1#2{\frac{\partial #1}{\partial #2}}
\newcommand{\F}{\mathbb{F}}

\DeclareMathOperator{\rank}{rank}

\newcommand{\actie}{\Phi}

\title{Aspects of reduction and transformation of Lagrangian systems with symmetry}
\author{ E. Garc\'{\i}a-Tora\~{n}o Andr\'{e}s$^{a,}$\footnote{Email: eduardo.gtoranoandres@ugent.be}, \ B. Langerock$^{a,b,c,}$\footnote{Email: bavo.langerock@ugent.be}\ and F. Cantrijn$^{a,}$\footnote{Email: frans.cantrijn@ugent.be}
\\[1.5\parskip]$^{a}$ Department of Mathematics, Ghent University\\ Krijgslaan 281 S22, B9000 Ghent, Belgium\\[1.5\parskip]
 $^{b}$ Department of Mathematics, KU~Leuven\\ Celestijnenlaan 200B, B3001 Leuven, Belgium\\[1.5\parskip]
$^{c}$ Belgian Institute for Space Aeronomy\\ Ringlaan 3, B1180 Brussels, Belgium\\
}
\date{}
\begin{document}
\maketitle\vspace{-.8cm}
\begin{abstract}
This paper contains results on geometric Routh reduction and it is a continuation of a previous paper~\cite{BEC} where a new class of transformations is introduced between Lagrangian systems obtained after Routh reduction. In general, these reduced Lagrangian systems have magnetic force terms and are singular in the sense that the Lagrangian does not depend on some velocity components. The main purpose of this paper is to show that the Routh reduction process itself is entirely captured by the application of such a new transformation on the initial Lagrangian system with symmetry.
\end{abstract}
{\em Key words:}\ {Symplectic reduction; Routh reduction; Lagrangian reduction; Reduction by stages}

{\em 2010 Mathematics Subject Classification:}\ 37J05; 37J15; 53D20

\section{Introduction}

A modern differential geometric treatment of Routh reduction for mechanical systems, as a Lagrangian analogue of Hamiltonian symplectic reduction, started in~\cite{doublependulum}. In that paper, a non-Abelian version of the classical reduction procedure of Routh was developed, thereby emphasizing the role of the magnetic or gyroscopic force term which appears after reduction. When taking this force term into account, the solutions of the Euler-Lagrange equations for the reduced Lagrangian, also called the Routhian, are the projections of those corresponding to the original Lagrangian. The definition of the Routhian involved the use of the {\sl mechanical connection}, which is the natural connection induced by the kinetic energy and the symmetry, and is the one leading to a reduced Lagrangian of mechanical type (see~\cite{lom}). This work received a natural continuation in~\cite{jalna,marsdenrouth} where, in particular, the variational aspects of the theory were further studied.

A key result in the further development of the theory is the realization of the reduced space on which the Routhian is defined as a fiber product (a result sometimes referred to as the {\sl realization theorem}, see~\cite{jalna,marsdenrouth}). This result was originally proved for mechanical systems, but remains valid for arbitrary Lagrangians on which an additional regularity condition is imposed, as shown in~\cite{BC}. The generalization for non-mechanical Lagrangians of the Routh technique was also studied in~\cite{mestcram} from the perspective of the Euler-Lagrange vector field. A different approach, from the point of view of exterior differential systems, is to be found in~\cite{Morando}.


The results presented in this paper are best situated within an ongoing research project related to the reduction theory of Lagrangian systems with symmetry and, in particular, to the technique of Routh reduction. In previous papers, different aspects of geometric Routh reduction have been studied. In~\cite{BC} the close relationship between Routh reduction and symplectic reduction was demonstrated, and this lead, on the one hand, to a broadening of the framework for applying Routh reduction by incorporating the so-called quasi-invariant Lagrangian systems and, on the other hand, to a description of Routh reduction by stages~\cite{routhstages}. In~\cite{BM} the regularity condition on the momentum map that is typical for Routh reduction, is relaxed. Then, the class of `magnetic Lagrangian systems' was introduced in~\cite{routhstages}, in the context of Routh reduction by stages, and its characteristic property is precisely that it is closed under the procedure of Routh reduction. Finally, in~\cite{BEC} a special class of transformations between magnetic Lagrangian systems was introduced and, by applying such a transformation, the magnetic Lagrangian system obtained after reduction of a Lagrangian system with respect to a full semi-direct product symmetry group could be identified with the system obtained after reduction by an Abelian subgroup.

The purpose of this paper is twofold. Firstly, we revisit this new concept of transformation in a more general framework and formalize its definition.  Secondly, we demonstrate that the process of Routh reduction may be understood as the result of two steps: the application of such a transformation to the initial Lagrangian system with symmetry, followed by a trivial reduction. This indicates the importance of these transformations and sheds some new light on the geometric structure underlying Lagrangian systems with symmetry.

The paper is organized as follows. In Section~\ref{sec:background} we recall the basics about magnetic Lagrangian systems and we give a description of Routh reduction in this framework. We then study transformations between magnetic Lagrangian systems in Section~\ref{sec:transf}, and describe a special class of transformations preserving some geometric properties of reduced spaces, among others the reduced presymplectic structure. We conclude Section~\ref{sec:transf} by giving the explicit transformation that corresponds to a restriction on the level set of the momentum map in the case of a standard Lagrangian system. In Section~\ref{sec:transrouth} we complete the picture by describing a general transformation for an arbitrary magnetic Lagrangian system, including a reduction step. Finally, in Section~\ref{sec:Ham} we briefly discuss an analogue of these transformations in a Hamiltonian framework.

\section{Background and notations}\label{sec:background}
\paragraph{\bf Group actions and bundles.} Throughout this paper $\actie^M: G\times M\to M$ denotes a left action of the Lie group $G$ on the manifold $M$, and we will use the shorthand notations $\actie^M(g,m)= gm=g\cdot m$. As usual, we let $\lag$ denote the Lie algebra of $G$ and $Ad$ the adjoint action of $G$ on $\lag$. The infinitesimal action of $\Phi^M$ on $\lag$, referred to as the infinitesimal generator map, is
\begin{align*}
\sigma^M_m:\lag & \to T_mM, \\
\xi & \mapsto \xi_M(m):=d/d\e|_0 (\exp \e\xi m),
\end{align*}
where $\exp:\lag \to G$ is the exponential map. $M$ can be naturally fibered over the orbit space $M/G=\{[m]_G:m\in M\}$ (here $[m]_G$ is the $G$-orbit through $m$) via $\pi:M\to M/G; m\to [m]_G$. It is well known that under the assumption that the action is free and proper, $\pi:M\to M/G$ has the structure of a principal $G$-bundle (see~\cite{koba}). From now on, unless otherwise stated, all actions are assumed to satisfy these requirements.

When a fiber bundle $\e:P\to Q$ is given, the fiber products $TQ\times_Q P$ and $T^*Q\times_Q P$ will be abbreviated by $T_PQ$ and $T^*_PQ$ respectively. We shall denote points on $T_PQ$ by $(v_q,p)$, where $v_q\in T_qQ$ and $p\in P$ is such that $\e(p)=q$ (and, in the same way, $(\alpha_q,p)$ denotes an arbitrary point in $T^*_PQ$). The contraction of an element $(\alpha_q,p)\in T^*_PQ$ with $(v_q,p)\in T_PQ$ is defined as $\langle (\alpha_q,p),(v_q,p) \rangle :=\langle \alpha_q,v_q\rangle$.
\begin{definition}  Let $\e:P\to Q$ be a fiber bundle.
\begin{enumerate}
\item $\tau_1:T_PQ \to TQ$ is the projection that maps $(v_q,p )\in T_PQ$ onto $v_q\in TQ.$
\item $\tau_2: T_PQ \to P$ is the projection that maps $(v_q,p)\in T_PQ$ onto $p\in P.$
\item $\pi_1:T_P^*Q \to T^*Q$ is the projection that maps $(\alpha_q,p )\in T_P^*Q$ onto $\alpha_q\in T^*Q.$
\item $\pi_2: T_P^*Q \to P$ is the projection that maps $(\alpha_q,p)\in T_P^*Q$ onto $p\in P.$
\end{enumerate}
\end{definition}
In agreement with the previous definition, when working with several bundles $\e^{(i)}:P_i\to Q_i$ we let $\tau_1^{(i)}$ and $\tau_2^{(i)}$ (respectively, $\pi_1^{(i)}$ and $\pi_2^{(i)}$) denote the corresponding projection maps of $T_{P_i}Q_i$ (resp. $T^*_{P_i}Q_i$).

The vertical tangent space to the fibration $\e$ at the point $p\in P$ is $V_p\e:=\ker{T_p\e}$. A connection $\A$ on the fiber bundle $\e:P\to Q$ is a $V\e$-valued 1-form $\A$ on $P$ such that $\A(v_p)=v_p$, for all $v_p$ in $V\e$. In the case of a principal fiber bundle $\e:P\to P/G$ with structure group $G$ a connection $\A$ determines a $\lag$-valued 1-form $\Ac$ on $P$ such that $\sigma^{P}\circ \Ac=\A$, and as a consequence $\Ac ( \xi_P)=\xi$, for all $\xi\in\lag$. A connection is called a principal connection if in addition $\A$ is equivariant, or $\big(\Phi^*_g\Ac\big)(v_p) = Ad_g \cdot \Ac(v_p)$, for all $g\in G$ and $v_p\in TP$. For any element $\mu\in\lag^*$, $\Ac_\mu$ denotes a 1-form on $P$ obtained by contraction of $\mu$ and the values of principal connection $\Ac$ on the level of the Lie algebra: $\Ac_\mu(v)=\langle \Ac(v),\mu \rangle $.

If we are given a chain of bundle structures $P\xrightarrow{\e_1} Q\xrightarrow{\e_2} R$, a connection on $\e_2:Q\to R$ induces a $V{\e_2}$-valued 1-form on $P$ defined by $\A_P(v_p):=\A(T_p\e_1(v_p))$. If, additionally, $\e_2:Q\to Q/G$ is a principal fiber bundle, we can construct a $\lag$-valued 1-form on $P$ as follows: $\Ac_P(v_p):=\Ac(T_p\e_1(v_p))$. Following the previous convention, we denote $\Ac_\mu(v)=\langle \Ac_P(v),\mu \rangle $ its contraction with a given $\mu\in\lag^*$.

\paragraph{Magnetic Lagrangian systems.} As mentioned in the introduction, magnetic Lagrangian systems appear naturally when applying Routh reduction to a Lagrangian system with symmetry. We now give a general definition of such a system.

\begin{definition} A magnetic Lagrangian system is a triple $(\e:P\to Q, L, \B)$ where $\e:P\to Q$ is a fiber bundle, $L$ is a smooth function on the fiber product $T_PQ$ and $\B$ is a closed 2-form on $P$.
\end{definition}
$P$ is thereby playing the role of configuration space of the system, $L$ is called the Lagrangian and $\B$ is referred to as the magnetic 2-form.  The standard notions of Legendre transformation and of energy, carry over to this new setting in a straightforward way:

\begin{definition} Let $(\e:P\to Q, L, \B)$ be a magnetic Lagrangian system. Then:
\begin{enumerate}
\item  The fiber derivative of $L$ is the map $\F L:T_PQ \to T_P^*Q$ sending $(v_q,p)\in T_PQ$ into $(\alpha_q,p)\in T^*_PQ$, where $\alpha_q\in T^*_qQ$ is (uniquely) determined by the relation
\[
\langle \alpha_q, w_q\rangle = \left.\frac{d}{du}\right|_{u=0} L(v_q+uw_q,p), \mbox{ for all } w_{q}\in T_{q}Q
\]
\item The energy is the function on $T_PQ$ defined by $E_L(v_q,p) =\langle \F L (v_q,p), (v_q,p)\rangle - L(v_q,p)$.
\end{enumerate}
\end{definition}

There is a natural way to construct a closed 2-form $\Omega^{L,B}$ on $T_PQ$ providing a generalization of the classical notion of Poincar\'e-Cartan 2-form, i.e.
\[
\Omega^{L,\B}:=\F L^*(\pi^*_1\omega_{Q} +\pi^*_2\B)\,.
\]
For later use, we give the coordinate expression of this 2-form. We will work with coordinates adapted to the fibration $\e$. If $(q^i)$ are local coordinates on $Q$, $(q^i,p^a)$ will denote bundle-adapted coordinates on $P$. Hence, on $T_PQ$ we have coordinates $(q^i,v^i,p^a)$ ($v^{i}$ denotes the ${i}$-th velocity component $\dot q^i$), and the Lagrangian $L$ is a function dependent of $(q^i,v^i,p^a)$. Straightforward computations show:
\begin{equation} \label{eq:symform}
\Omega^{L,\B} = d\left(\fpd{L}{v^i} \right)\wedge d q^i + \frac12 \B_{ij} dq^i \wedge dq^j + \B_{ia}dq^i\wedge dp^a + \frac12 \B_{ab} dp^a\wedge dp^b.
\end{equation}
In the case where the 2-form $\pi^*_1\omega_{Q} +\pi^*_2\B$ on $T^*_PQ$ is symplectic and $\F L$ is a (global) diffeomorphism, we say that the magnetic Lagrangian system $(\e:P\to Q, L, \B)$ is (hyper)regular. This guarantees that $\Omega^{L,B}$ is a symplectic form.

A curve $p(t)\in P$ induces a curve on $T_PQ$, namely $\gamma(t):=(\dot q(t), p(t))$, where $q(t) = \epsilon(p(t))$. The curve $p(t)$ is said to be a solution of the Euler-Lagrange (EL) equations iff $\gamma(t)$ satisfies the equation \begin{equation}\label{eq:EL} i_{\dot \gamma(t)} \Omega^{L,\B} (\gamma(t)) = -dE_L(\gamma(t)).\end{equation} The reader is referred to~\cite{routhstages} for a coordinate expression of these EL equations. We conclude with two remarks:
\begin{itemize}
\item The above definition of a magnetic Lagrangian system incorporates the standard concept of a Lagrangian system on a manifold $Q$ by letting $P=Q$, $\e$ the identity and $\B=0$.
\item The Lagrangian $L$ of any magnetic Lagrangian system can be pulled-back to a function on $TP$ and determines there a standard Lagrangian system with a magnetic force term, whose EL equations are precisely \eqref{eq:EL}. Since this $L$ on $TP$ does not depend on all velocity components, it is a singular Lagrangian by construction.
\end{itemize}
Hyperregular magnetic Lagrangian systems, although singular from a classical point of view, are amenable to a symplectic description, and it is therefore possible to study their dynamics using a Hamiltonian framework.

\paragraph{Routh reduction.} Assume a hyperregular standard Lagrangian system $(P=Q,L,\B=0)$ is given, and let $\Omega_L:=\F L^*\omega_Q$ be its Poincar\'e-Cartan symplectic form. Let $\Phi^Q=\Phi$ be a free and proper $G$-action on $Q$ such that $L$ is invariant with respect to its tangent lift $T\Phi$, and let $J_L$ be the equivariant momentum map $J_L : TQ\to \lag^*$ for $T\Phi$:
\[
\langle J_L(v_q),\xi\rangle = \langle \F L(v_q) , \xi_{Q}(q)\rangle.
\]
We then know that $J_{L}$ is constant along solutions of the EL equations. In particular, fixing a regular value $\mu\in\lag^*$ of the momentum map, $X_{E_L}$ leaves $J_L^{-1}(\mu)$ invariant. Moreover, equivariance of $J_L$ implies that the action of $G_\mu$ on $TQ$ restricts to a (free and proper)  action on $J_L^{-1}(\mu)$. It can be shown that, under a certain regularity condition on the Lagrangian $L$, we can realize the orbit space $J_L^{-1}(\mu)/G_{\mu}$ as the fiber product $T_{Q/G_{\mu}}(Q/G)$. It should be noted that the condition of $\mu$ being a regular value is redundant once freeness of the action is assumed.
\begin{definition} A $G$-invariant Lagrangian $L$ is called~\emph{$G$-regular} if for every fixed $v_q \in TQ$ the map $\lag \to \lag^*, \xi \mapsto J_L(v_q+\xi_Q(q))$ is a diffeomorphism.
\end{definition}
From now on, we will assume that $L$ is $G$-regular. Making use of the $G$-regularity, one can then construct a diffeomorphism $\Pi_{\mu}:J_L^{-1}(\mu)/G_{\mu}\to T_{Q/G_{\mu}}(Q/G)$ (see \cite{BEC}, Lemma~1).

Next, fix a connection $\A$ in the principal bundle $\pi: Q \rightarrow Q/G$, and let, as before, $\Ac$ be the associated $\lag$-valued 1-form. The 2-form $d\Ac_\mu$ is easily checked to be projectable to a 2-form $\mathfrak{B}_\mu$ on $Q/G_\mu$. Define the following function on $TQ$:
\[
R^\mu=L-\Ac_\mu,
\]
(although $\Ac_\mu$ is a 1-form, in the above definition it is understood to be the function $\Ac_\mu: TQ \to \mathbb{R}; v_q \mapsto  \Ac_\mu (q)(v_q )$). Due to $G_\mu$-invariance, the restriction of $R^\mu$ to $J_{L}^{-1}(\mu)$ induces a function on the quotient $J^{-1}_L(\mu)/G_\mu\cong T_{Q/G_\mu}(Q/G)$. $\Ro^\mu$ denotes the corresponding function on $T_{Q/G_\mu}(Q/G)$ and is called the {\em Routhian}.

\begin{proposition}[Routh reduction]\label{prop:normalrouth} Let $L$ be a hyperregular $G$-invariant, $G$-regular Lagrangian with configuration space $Q$, and let $\mu \in \lag^*$ denote a regular value of the momentum map $J_L$. Then, the magnetic Lagrangian system $(Q/G_\mu\to Q/G, \Ro^{\mu},\mathfrak{B}_\mu)$, as constructed above, has the property that every solution of the original Euler-Lagrange equations corresponding to the momentum value $\mu$ projects onto a solution of $(Q/G_\mu\to Q/G, \Ro^{\mu},\mathfrak{B}_\mu)$. Conversely, every solution in $Q/G_\mu$ of  $(Q/G_\mu\to Q/G, \Ro^{\mu},\mathfrak{B}_\mu)$ is the projection of a solution to the Euler-Lagrange equations for $L$ with momentum $\mu$.
\end{proposition}

\paragraph{Example.} Consider three planar rigid bodies with a common fixed point $O$, so that each body is free to rotate about the axis through $O$, orthogonal to the plane. The configuration space is $S^1 \times S^1\times S^1$, with coordinates  $(\theta,\varphi,\psi)$ where $\theta$ is the angle which the first body makes with a fixed direction in the plane, $\varphi$ is the relative angle of the second rigid body w.r.t. the first and finally $\psi$ denotes the relative angle of the third rigid body w.r.t. the second (see Figure~\ref{fig:ex}).
\begin{figure}[h]
\begin{center}
\includegraphics{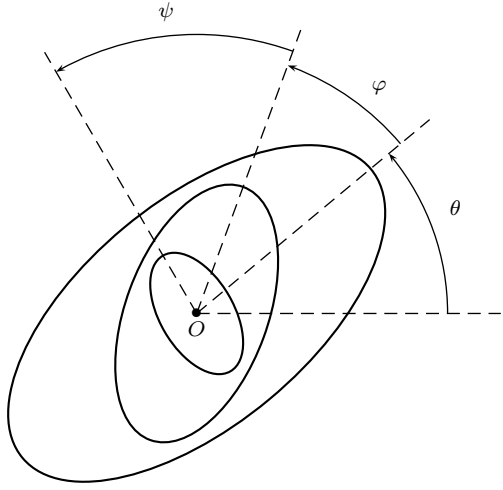}
\end{center}
\caption{Coordinates on $S^1\times S^1\times S^1$}\label{fig:ex}
\end{figure}


The potential is supposed to be of the form $V(\varphi,\psi)$ such that we have an $S^1$-invariant Lagrangian:
\[
L = \frac{1}{2}I_{1}\dot{\theta}^{2}+\frac{1}{2}I_2(\dot{\theta}+\dot{\varphi})^{2}
+\frac{1}{2}I_3(\dot{\theta}+\dot{\varphi}+\dot{\psi})^{2}-V(\varphi,\psi),
\]
whose EL equations (in normal form) are:
\[
\ddot{\theta} =  \frac{1}{I_1}\fpd{V}{\varphi}, \quad \ddot{\varphi} = -\left(\frac{I_1+I_2}{I_2I_1}\right)\fpd{V}{\varphi}+\frac{1}{I_2}\fpd{V}{\psi}, \quad \ddot{\psi} =  -\left(\frac{I_2+I_3}{I_3I_2}\right)\fpd{V}{\psi}+\frac{1}{I_2}\fpd{V}{\varphi}.
\]
Fix a regular value $\mu$ for the momentum $J=(I_1+I_2+I_3)\dot{\theta}+(I_2+I_3)\dot{\varphi}+I_3\dot{\psi}$. Then on the level set $\{J=\mu\}$ we have $\dot{\theta}=(\mu-(I_2+I_3)\dot{\varphi}+I_3\dot{\psi})/(I_1+I_2+I_3)$. We work out the Routhian for two different connections:
\begin{enumerate}[1)]
\item  Consider the {\sl mechanical connection} $\A^M$ whose horizontal spaces are orthogonal to the $G$-orbits with respect to the metric given by the kinetic energy. It is easy to check that one has the following connection 1-form: $\Ac^M=d\theta+\frac{I_2+I_3}{I_1+I_2+I_3}d\varphi+\frac{I_3}{I_1+I_2+I_3}d\psi$. The Routhian $R^\mu_M(\varphi,\dot\varphi,\psi,\dot\psi)=\big(L-
\langle\Ac^M,\mu\rangle\big)_{\{J=\mu\}}$ satisfies:
\begin{align*}
R^\mu_M\simeq & \frac{1}{2}\left[\frac{I_1\left(I_2+I_3\right)}{\left(I_1+I_2+I_3\right)}\right]\dot{\varphi}^{2}+
\frac{1}{2}\left[\frac{I_3\left(I_1+I_2\right)}{\left(I_1+I_2+I_3\right)}\right]\dot{\psi}^{2}\\
&+\left[\frac{I_1I_3}{\left(I_1+I_2+I_3\right)}\right]\dot{\varphi}\dot{\psi}-V(\varphi,\psi),
\end{align*}
where the symbol $\simeq$ means that we have omitted constant terms. Note that with this choice of the connection, the Routhian is again of mechanical type.
\item Take now the non-flat connection given by $\Ac^0=d\theta+\cos(\psi)d\varphi$. The Routhian $R^\mu_0(\varphi,\dot\varphi,\psi,\dot\psi)=\big(L-
\langle\Ac^0,\mu\rangle\big)_{\{J=\mu\}}$ satisfies:
\begin{align*}
R^\mu_0\simeq & \frac{1}{2}\left[\frac{I_1\left(I_2+I_3\right)}{\left(I_1+I_2+I_3\right)}\right]\dot{\varphi}^{2}+
\frac{1}{2}\left[\frac{I_3\left(I_1+I_2\right)}{\left(I_1+I_2+I_3\right)}\right]\dot{\psi}^{2}
+\left[\frac{I_1I_3}{\left(I_1+I_2+I_3\right)}\right]\dot{\varphi}\dot{\psi}\\
&+\frac{\mu\left(I_2+I_3\right)(1-\cos(\psi))}{\left(I_1+I_2+I_3\right)}\dot{\varphi}+\frac{\mu I_3}{\left(I_1+I_2+I_3\right)}\dot{\psi}-V(\varphi,\psi).
\end{align*}
An easy computation shows that the EL equations for any of $R^\mu_M$ or $R^\mu_0$ are equivalent to the EL equations for the variables $(\varphi,\psi)$ of $L$ (note that the EL equations for $R^\mu_0$ have a force term $d\Ac^0_\mu=\mu\sin(\psi)d\varphi\wedge  d\psi$). Together with the momentum equation, they provide complete solutions of the original system.
\end{enumerate}
This example illustrates an important fact about Routh reduction: the choice of the connection is arbitrary and always leads to the same EL equations.

\section{Transformations between magnetic Lagrangian systems}\label{sec:transf}

In this section we study transformations between magnetic Lagrangian systems. As described in the introduction, the main goal is to obtain the Routh reduction procedure as a transformation between magnetic Lagrangian systems. Throughout this section, we develop Routh reduction as an example of the general theory on these transformation.

\subsection{Pull-back Hamiltonian systems}\label{subset:pullback}
We first recall some generalities concerning the pull-back of a symplectic structure and investigate the relationship between Hamiltonian vector fields that are connected by such a pull-back operation.

Consider the situation where we are given two manifolds $N,M$ and a smooth map $f:N\to M$ of constant rank. Assume, in addition, that $M$ is a symplectic manifold with symplectic form $\omega_M$, and that we are given a Hamiltonian function $h_M$ on $M$. Let us denote by $X_{h_M}$ the corresponding Hamiltonian vector field, which satisfies $i_{X_{h_M}}\omega_M=-dh_M$. Consider then the presymplectic form $\omega_N = f^*\omega_M$ and the Hamiltonian function $h_N = f^*h_M$, induced on $N$. A Hamiltonian vector field on $N$ with respect to $\omega_N$, corresponding to $h_{N}$, is determined by the presymplectic equation
\begin{equation} \label{eq:presymp} i_{X} \omega_N = -dh_N,
\end{equation}
and we are interested in those cases where (some of) the integral curves of $X_{h_M}$ can be retrieved from integral curves of a solution to \eqref{eq:presymp}. More precisely, we investigate when $X_{h_{M}}$ is $f$-related to a solution $X$ of \eqref{eq:presymp}. Recall that solutions to \eqref{eq:presymp}, if they exist, are determined up to elements in the kernel of $\omega_{N}$, which we denote by $TN^{\omega_{N}}$, and that $Tf\left(TN^{\omega_N}\right)= [Tf(TN)]^{\omega_M}\cap Tf(TN)$, where $[Tf(TN)]^{\omega_M}$ is the kernel of the restriction of $\omega_M$ to $TM_{|f(N)}$.

First note that any vector field $Y$ on $N$ which is $f$-related to $X_{h_M}$, solves \eqref{eq:presymp}: for any $x \in N$ and $Z_x\in T_xN$ it follows that
\begin{align*}
\omega_N(x)\big(Y_x,Z_x\big)&=\omega_M(f(x))\big(Tf(Y_x),Tf(Z_x)\big) = -dh_M(f(x))\big(Tf(Z_x)\big)\\
&=-dh_N(x)(Z_x).
\end{align*}
A necessary condition for $X_{h_{M}}$ to be $f$-related to a vector field on $N$ is ${X_{h_M}}_{|f(N)}\in Tf(TN)$, or equivalently
\begin{equation} \label{eq:tan}
\left\langle dh_M,\left[Tf(TN)\right]^{\omega_M}\right\rangle _{|f(N)} = 0.
\end{equation}
If $X$ solves \eqref{eq:presymp} and condition \eqref{eq:tan} holds, the vector $Tf(X_x) - X_{h_{M}}(f(x))$ is in $Tf(T_xN^{\omega_{N}})$ for all $x$ in the domain of $X$, i.e. $X$ can be gauged by an element in the kernel of $\omega_{N}$ so that it becomes $f$-related to $X_{h_{M}}$. To show that \eqref{eq:tan} is also a sufficient condition for the existence of an $f$-related solution of \eqref{eq:presymp}, we need to show that it implies the existence of a solution \eqref{eq:presymp}. For that purpose, we rely on the presymplectic constraint algorithm developed by M. Gotay, J.M. Nester and G. Hinds (see~\cite{Cons_Hinds,gotaya}).

The starting point of the presymplectic constraint algorithm is the observation that \eqref{eq:presymp} admits a solution at a point $x \in N$ if the following condition holds: $\langle Z_x,dh_N(x)\rangle = 0$ for all $Z\in TN^{\omega_N}$. The set of all these points is assumed to form a (immersed) submanifold of $N$, i.e.
\[
N_2 = \{x \in N \,:\, \langle dh_N(x), T_xN^{\omega_N}\rangle = 0\}\,,
\]
called the secondary constraint submanifold. The next step then consists in requiring that one should be able to find a vector field solution to \eqref{eq:presymp} which is tangent to $N_2$. This possibly leads to new constraints defining a constraint submanifold $N_3 = \{x \in N_2 \,:\, \langle dh_{N}(x), T_xN_2^{\omega_N} \rangle = 0\}$, where $TN_2^{\omega_N}= \{X \in TN_{|N_2}\,:\, \omega_N(X,Y) = 0\;\mbox{for all}\; Y\in TN_2\}$. Proceeding this way one generates a descending sequence of constraint submanifolds
$\ldots \subset N_k \subset \ldots \subset N_2 \subset N:=N_1$, where
\[
N_k = \{x \in N_{k-1} \,:\, \langle dh_N(x), T_xN_{k-1}^{\omega_N}\rangle = 0\}
\]
for $k=2,\ldots$, with $TN_{k-1}^{\omega_N} =\{X \in TN_{|N_{k-1}}\,:\, \omega_N(X,Y) = 0\;\mbox{for all}\; Y\in TN_{k-1}\}$.  If this sequence stabilizes at some finite step $K \in \mathbb{N}$, in the sense that $N_{K} \neq \emptyset$ and $N_{K+1} = N_K$, we say that $N_K$ is the \emph{final constraint (sub-)manifold}.
In that case, equation~\eqref{eq:presymp} admits solutions on $N_K$, and we say that the presymplectic equation leads to a consistent dynamics on $N_K$.

Returning to the situation described above, we are now able to prove that \eqref{eq:presymp} admits a consistent dynamics on $N$ provided the Hamiltonian vector field $X_{h_M}$ is everywhere tangent to $f(N)$. In fact we have:

\begin{proposition} There exists a solution $X$ of \eqref{eq:presymp} which is $f$-related to $X_{h_{M}}$ if and only if  ${X_{h_M}}_{|f(N)} \in Tf(TN)$.
\end{proposition}
\begin{proof} It suffices to check the first step of the presymplectic constraint algorithm. Indeed, from $Tf\left(TN^{\omega_N}\right)\subset [Tf(TN)]^{\omega_M}$ and using equation~\eqref{eq:tan} it follows that for all $x \in N$
\[
\langle dh_N (x), T_xN^{\omega_N} \rangle = \langle dh_M(f(x)), T_xf\left(T_xN^{\omega_N}\right) \rangle=0\,,
\]
proving that $N$ is the final constraint manifold for \eqref{eq:presymp} which therefore admits a solution. Hence, according to a previous observation, there also exists a solution which is $f$-related to $X_{h_M}$.
\end{proof}

\subsection{Compatible transformations}\label{subsec:compatibility}

We now specialize the symplectic framework given above to the case of interest in the study of magnetic Lagrangian systems, namely fiber products with (pre)symplectic structures of the form $\Omega^{L,\B}$.

\begin{definition} Let $\e^{(1)}:P_1\to Q_1$ and $\e^{(2)}:P_2\to Q_2$ be two fiber bundles. If $F:P_1\to P_2$ and $f:Q_2\to Q_1$ are two surjective submersions we say that the pair $(F,f)$ forms a \emph{transformation pair} between both bundles if the following equality holds:
\[
f \circ \e^{(2)} \circ F = \e^{(1)},
\]
and all the arrows in Figure~\ref{fig:diagram1} represent fiber bundles.
\begin{figure}[h]
\begin{center}
\includegraphics{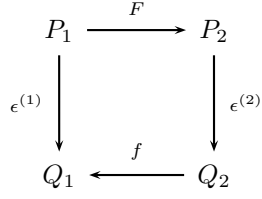}
\end{center}
\caption{Transformation pair}\label{fig:diagram1}
\end{figure}
\end{definition}

We write $\dim Q_i = n_i$ and $\dim P_i = n_i + k_i$ for $i=1,2$. Because $F$ and $f$ are submersions, it follows that $n_1 + k_1 \geq n_2 + k_2$ and $n_1 \leq n_2$. This way we find the relation $k_1 \geq k_2$ between the dimensions of the fibers of the bundles $\e^{(i)}:P_i\to Q_i$. A transformation pair induces a chain of bundle structures $P_1\to P_2\to Q_2\to Q_1$. Choosing coordinates adapted to these fibrations, we let $(q^i)$ denote coordinates on $Q_1$, $(q^i,\bar q^a)$ coordinates on $Q_2$, $(q^i,\bar q^a,\bar p^\alpha)$ on $P_2$ and finally $(q^i,\bar q^a,\bar p^\alpha,p^\gamma)$ on $P_1$. We have then the following natural sets of coordinates: $(q^i,v^i, \bar q^a,\bar p^\alpha, p^\gamma)$ on $T_{P_1}Q_1$  and $(q^i,\bar q^a,v^i,\bar v^a,\bar p^\alpha)$ on $T_{P_2}Q_2$.

\begin{definition} Let $(F,f)$ be a transformation pair between $\e^{(1)}:P_1\to Q_1$ and $\e^{(2)}:P_2\to Q_2$. Then:
\begin{enumerate}
\item  Two points $(v_{q_i},p_i)\in T_{P_i}Q_i$, $i=1,2$ are {\em $(F,f)$-compatible} if $F(p_1)=p_2$ and $Tf(v_{q_2})=v_{q_1}$.
\item A smooth map $\psi: T_{P_1}Q_1 \to T_{P_2}Q_2$ is {\em compatible} with the transformation pair $(F,f)$ if for every point $s_{1}=(v_{q_{1}},p_{1})\in T_{P_1}Q_1$, the points $s_{1}$ and $\psi(s_{1})$ are $(F,f)$-compatible.
\end{enumerate}
\end{definition}

We simply say that $\psi$ is a {\em compatible transformation} or {\em compatible map}. Compatibility for a map $\psi$ is equivalently specified by the following two conditions:
\begin{enumerate}[i)]
\item $\tau^{(2)}_2 \circ \psi = F \circ \tau^{(1)}_2$;
\item $Tf \circ \tau^{(2)}_1 \circ \psi = \tau^{(1)}_1$.
\end{enumerate}
The situation is summarized in Figure~\ref{fig:compa}:
\begin{figure}[H]
\begin{center}
\includegraphics{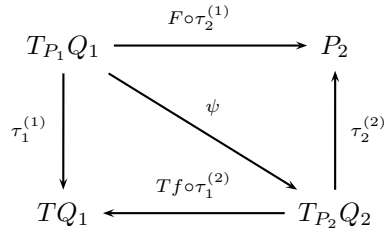}
\end{center}
\caption{Commutative diagram for a compatible map $\psi$}\label{fig:compa}
\end{figure}
We use coordinates adapted to the fibrations as introduced before to describe both a point and its image by $\psi$. It is then readily checked
that compatible maps convey to the following coordinate expression:
\[
\psi(q^i,v^i, \bar q^a,\bar p^\alpha, p^\gamma) = (q^i,\bar q^a,v^i,\bar v^a = \psi^a(q^i,v^i, \bar q^a,\bar p^\alpha, p^\gamma), \bar p^\alpha).
\]
Note that the rank of a transformation $\psi$ is determined by the rank of the matrix $(\partial\psi^a/\partial p^\gamma)_{a,\gamma}$ in the following way: $\rank \psi=\dim P_2+\dim Q_1+\rank (\partial\psi^a/\partial p^\gamma)_{a,\gamma}$. In particular, for $\psi$ to be a diffeomorphism the dimension of the fibers corresponding to $f$ and $F$ must the same and $\det(\partial\psi^a/\partial p^\gamma)_{a,\gamma}\neq0$.

The compatibility of points gives naturally a notion of compatibility of vectors by lifting the conditions to the tangent spaces:

\begin{definition} Let $(F,f)$ be a transformation pair between $\e^{(1)}:P_1\to Q_1$ and $\e^{(2)}:P_2\to Q_2$, and let $s_{1}=(v_{q_1},p_1)\in T_{P_1}Q_1$ and $s_{2}=(v_{q_2},p_2)\in T_{P_2}Q_2$ be arbitrary points. Given $Y_{s_{1}}$ and $X_{s_{2}}$ tangent vectors at $s_{1}$ and $s_{2}$ respectively, we say that $Y_{s_{1}}$ and $X_{s_{2}}$ are {\em $(F,f)$-compatible} if the following two conditions are satisfied:
\begin{enumerate}
\item $T\big(F\circ \tau^{(1)}_2\big)(Y_{s_{1}}) = T\tau^{(2)}_2 (X_{s_{2}})$;
\item $T\tau^{(1)}_1(Y_{s_{1}}) = T\big(Tf \circ \tau^{(2)}_1\big)(X_{s_{2}})$.
\end{enumerate}
\end{definition}
Note that, in particular, ${s_{1}}$ and $s_{2}$ need to be compatible points (see Figure~\ref{fig:diagram2}).
\vspace{.5cm}
\begin{figure}[thb]
\begin{center}
\includegraphics{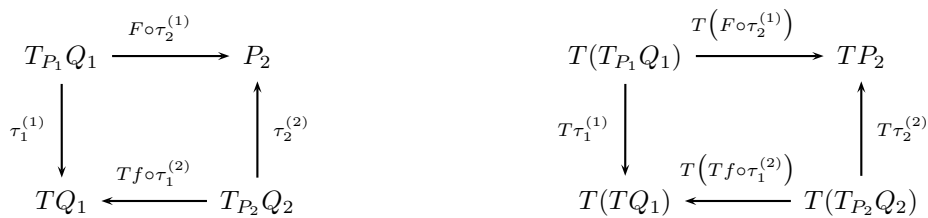}
\end{center}
\caption{Compatible points and vectors}\label{fig:diagram2}
\end{figure}

Consider an arbitrary vector $Y_{s_{1}}$ tangent to $T_{P_1}Q_1$ at the point $s_{1}$. Its coordinate expression is
\begin{equation}\label{eq:compavecy}
Y_{s_{1}} = Y_{s_{1}}^i \fpd{}{q^i} + Y_{s_{1}}^a \fpd{}{\bar q^a} + Y_{s_{1}}^\alpha \fpd{}{\bar p^\alpha} + Y_{s_{1}}^\gamma
\fpd{}{p^\gamma} + \hat Y^i_{s_{1}} \fpd{}{v^i},
\end{equation}
and reading the local expressions of the previous definition, a compatible tangent vector $X_{s_{2}}$ at the compatible point ${s_{2}}$ assumes the following form:
\begin{equation}\label{eq:compavecx}
X_{s_{2}}= Y_{s_{1}}^i \fpd{}{q^i} + Y_{s_{1}}^a \fpd{}{\bar q^a} + Y_{s_{1}}^\alpha \fpd{}{\bar p^\alpha}  + \hat Y^i_{s_{1}} \fpd{}{v^i} + \hat X_{s_{2}}^a\fpd{}{\bar v^a}.
\end{equation}

Given a compatible transformation $\psi$ between $\e^{(1)}:P_1\to Q_1$ and $\e^{(2)}:P_2\to Q_2$, it is clear that $Y_{s_{1}}$ and $X_{s_{2}} = T\psi(Y_{s_{1}})$ are compatible vectors for any $Y_{s_{1}}\in T_{P_1}Q_1$. In this particular case, from the coordinate expression of a compatible map, we find:
\[
\hat X_{s_{2}}^a = Y_{s_{1}}^i\fpd{\psi^a}{q^i} + \hat Y^i_{s_{1}}\fpd{\psi^a}{v^i} + Y_{s_{1}}^a\fpd{\psi^a}{\bar q^a} + Y_{s_{1}}^\alpha\fpd{\psi^a}{\bar p^\alpha} + Y_{s_{1}}^\gamma\fpd{\psi^a}{p^\gamma}.
\]
\paragraph{Example (Routh reduction).} Consider a hyperregular standard Lagrangian system $(Q\to Q,L,\B=0)$ amenable to Routh reduction, i.e. there is a left $G$-action and $L$ is $G$-invariant and $G$-regular. Consider the (trivial) bundles $\e^{(2)}=\mbox{id}_{Q} : P_{2}=Q\to Q_{2}=Q$ and $\e^{(1)}=\pi: P_{1}=Q\to Q_{1}=Q/G$ (Figure~\ref{fig:diagramTQ}). The maps $F=\mbox{id}_{Q}:P_{1}\to P_{2}$ and $f=\pi: Q_{2}=Q\to  Q_{1}=Q/G$ are a transformation pair between $\e^{(1)}$ and $\e^{(2)}$. Then $T_{P_{1}}Q_{1} = T_{Q}(Q/G)$, $T_{P_{2}}Q_{2} = TQ$ and it follows:
\begin{itemize}
\item points $(v_{[q]_{G}},q)$ and $v_{q}$ in $T_{Q}(Q/G)$ and $TQ$ respectively, are compatible if $v_{[q]_{G}} = T\pi(v_{q})$;
\item a map $\psi: T_{Q}(Q/G) \to TQ$ is compatible if it sends $(v_{[q]_{G}}, q)$ to a tangent vector in $TQ$ projectable to $v_{[q]_{G}}$, i.e. the map is determined up to a gauge in $\lag$;
\item tangent vectors $X\in T(TQ)$ and $Y=(Y^{Q},Y^{T(Q/G)})\in T(T_{Q}(Q/G))\cong TQ\times_{Q}T(T(Q/G))$ are compatible if  $T\tau_{Q}(X)=Y^{Q}$ and $T(T\pi)(X)=Y^{T(Q/G)}$.
\end{itemize}
\begin{figure}[thb]
\begin{center}
\includegraphics{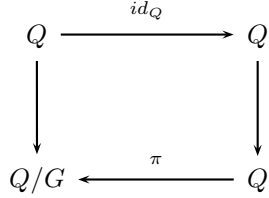}
\end{center}
\caption{Routh reduction scheme in $TQ$}\label{fig:diagramTQ}
\end{figure}

\paragraph{A family of compatible transformations.} Assume we are given two fiber bundles $\e^{(i)}: P_i \rightarrow Q_i$, $i=1,2$, and a magnetic Lagrangian system  $(\e^{(2)},L_2,\B_2)$, together with a transformation pair $(F,f)$ between $\e^{(1)}:P_1\to Q_1$ and $\e^{(2)}:P_2\to Q_2$. We can then construct a family of compatible transformations $\psi_{L_2,\beta}:T_{P_1}Q_1 \to T_{P_2}Q_2$ between these spaces. As the notation suggests, this family depends on the Lagrangian $L_{2}$ and an arbitrary map $\beta:P_{1}\to V^{*}f$, where $V^{*}f$ is the dual of the bundle $Vf$ of tangent vectors vertical to the fibration $f$.

First we introduce the notion of $f$-regularity of the Lagrangian $L_2$.
Consider the map $\alpha_{L_2}: T_{P_2}Q_2 \to V^*f$ which is defined as he composition of $\pi^{(2)}_1 \circ {\F L}_2: T_{P_2}Q_2 \to T^*Q_2$ with the projection of $T^*Q_2$ onto $V^*f$.
\begin{definition} The Lagrangian $L_2$ is {\em $f$-regular} if for any given $s_{2}=(v_{q_2},p_2) \in T_{P_2}Q_2$ the map $$\alpha^{s_{2}}_{L_2}: V_{q_2}f \to V^*_{q_2}f; w_{q_2} \mapsto \alpha_{L_2}(v_{q_2} + w_{q_2},p_2)$$ is a diffeomorphism.
\end{definition}
It is easily verified in coordinates that this condition is equivalent to the non-vanishing of the Hessian of $L_2$ with respect to the velocities, i.e.
\[
\det\left(\frac{\partial^2 L_2}{\partial {\bar  v}^a \partial {\bar  v}^b}\right) \neq 0.
\]
For $f$-regular Lagrangians, we are now ready to introduce a family of compatible maps $\psi_{L_2,\beta}:T_{P_1}Q_1\to T_{P_2}Q_2$.
\begin{enumerate}[I)]
\item Consider the map $\alpha_{L_2}: T_{P_2}Q_2 \to V^*f$, defined as above;
\item Fix a map $\beta: P_1 \to V^*f$ such that $f\circ pr_{|V^*f}\circ\beta=\e^{(1)}$, where $pr: T^{*}Q_{2}\to Q_{2}$ denotes the standard projection on the cotangent bundle $T^{*}Q_{2}$ (see also Figure~\ref{fig:diagrambeta});
\begin{figure}[htb]
\begin{center}
\includegraphics{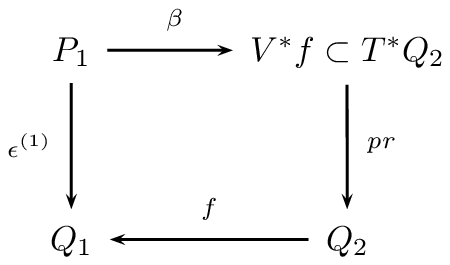}
\end{center}
\caption{Commutative diagram for the map $\beta$}\label{fig:diagrambeta}
\end{figure}
\item Let $s_{1}=(v_{q_1},p_1)$ be an arbitrary point in $T_{P_1}Q_1$  and let $s_{2}=(v_{q_2},p_2)\in T_{P_2}Q_2$ be a compatible point (such a point always exists).  Due to the $f$-regularity of $L_2$, there exists a unique tangent vector $w_{q_2}\in V_{q_2}f$ that satisfies $\alpha^{s_{2}}_{L_2}(w_{q_2}) = \beta(p_1)$, or alternatively
\[
\pi^{(2)}_1\big({\F L}_2{(v_{q_2}+w_{q_2},p_2)}\big)_{|Vf}= \beta(p_1).
\]
We take the point $(v_{q_2}+w_{q_2},p_2)$ as the image of $s_{1}=(v_{q_1},p_1)$ under $\psi_{L_2,\beta}$. The fact that $L_2$ is $f$-regular implies that the construction  is independent of the choice of $s_{2}$.
\end{enumerate}
By construction, the map $\psi_{L_2,\beta}$ is compatible and it satisfies $\alpha_{L_2} \circ \psi_{L_2,\beta} \equiv \beta \circ \tau_2^{(1)}$. In coordinates, writing $\beta=\beta_a d\bar q^a$, this relation takes the form:
\[
\fpd{L_2}{{\bar  v}^a}\bigg(q^i,\bar q^a,\bar p^\alpha,v^i,{\bar v}^a=\psi^a(q^i,\bar q^a,\bar p^\alpha,p^\gamma,v^i)\bigg) \equiv \beta_a(q^i,\bar q^a,\bar p^\alpha,p^\gamma).
\]

\begin{proposition}\label{pro:charac} $\psi_{L_2,\beta}$ is uniquely characterized by the following two conditions:
\begin{enumerate}
\item It is a compatible transformation;
\item It satisfies $\alpha_{L_2} \circ \psi_{L_2,\beta} \equiv \beta \circ \tau_2^{(1)}$.
\end{enumerate}
\begin{proof} Take $s_{1}\in T_{Q_1}P_1$ and let $\psi$ be a compatible map  that satisfies the relation above. This condition reads:
\[
\fpd{L_2}{{\bar  v}^a}\left(\psi(s_{1})\right) = \beta_a\big(\tau_2^{(1)}(s_{1})\big).
\]
If we use regularity of $\alpha_{L_2}$ and apply the inverse function theorem, it follows that $\psi$ is unique.
\end{proof}
\end{proposition}

\paragraph{Example (Routh reduction).} Recall the setup for a standard Lagrangian system on $Q$ amenable to Routh reduction:  $\e^{(1)}=\pi: P_{1}=Q\to Q_{1}=Q/G$, $\e^{(2)}=\mbox{id}_{Q} : P_{2}=Q\to Q_{2}=Q$, $F=\mbox{id}_{Q}$, $f=\pi: Q\to Q/G$ and $L_{2}=L$, $\B_{2}=0$. The bundle $Vf$ is the bundle of symmetry vectors $\{\xi_{Q_{2}} | \xi \in \lag\}$ and $\sigma^{*}\circ\alpha_{L} = J_{L}$. The map $\beta : Q \to V^{*}f \cong Q \times \lag^{*}$ is equivalent to a $\lag^{*}$-valued map on $Q$. Although we are running ahead of things, in the case of Routh reduction the map $\beta$ is determined from a fixed value $\mu\in\lag^{*}$. Indeed $\beta : Q\to V^*\pi $ is characterized in the following way:
\[
\langle \beta(q), \xi_Q(q)\rangle  = \langle \mu, \xi\rangle, \mbox{ for all } \xi\in\lag.
\]
Thus, the second equation in Proposition~\ref{pro:charac} coincides with the momentum equation $J_{L_{2}} = \mu$. From the definition of the map $\alpha_L$, the image of $(v_{[q]_G},q)$ by the map $\psi_{L,\beta}$ is the element $(v_q+\eta_Q)\in TQ$ with $\eta$ is determined by the equality:
\[
\langle \F L(v_q+\eta_Q), \xi_Q(q)\rangle  = \langle \beta(q), \xi_Q(q)\rangle  = \langle \mu, \xi\rangle,
\]
for all $\xi\in\lag$. From the definition of $J_{L}$ it follows $\langle \F L(v_q+\eta_Q), \xi_Q(q)\rangle =  \langle J_L(v_q+\eta_Q),\xi\rangle$ and hence $\psi_{L,\beta}=\imath_\mu\circ\Pi_{\mu}^{-1}$.

\paragraph{Pull-back of a magnetic Lagrangian system under $\psi_{L_{2},\beta}$.} In the next two paragraphs we study the pull-back  under $\psi_{L_{2},\beta}$ of the (pre)symplectic system $\Omega^{L_2,\B_2}$ with energy (Hamiltonian) $E_{L_2}$. In the first paragraph we show that the pull-back system is associated to a new magnetic Lagrangian system on $P_{1}\to Q_{1}$. In the second paragraph we study conditions on the map $\beta$ such that the EL equations of the pull-back system are related to the EL equations of the initial magnetic Lagrangian system.

In order to define in an intrinsic way a Lagrangian on $P_{1}\to Q_{1}$ whose associated 2-form equals $\psi_{L_2,\beta}^*\Omega^{L_2,\B_2}$, we choose a connection $\A$ on the bundle $f: Q_2 \to Q_1$. Recall from the introduction that $\A$ is a $Vf$-valued 1-form on $Q_2$, satisfying $\A(v_{q_2})=v_{q_2} $, for all $v_{q_2}\in Vf$. Consider now the associated $Vf$-valued 1-form $\A_{P_1}$ on $P_1$ defined from $\A_{P_1}(v)=\A\big(T(\e^{(2)} \circ F)(v)\big)$ for $v\in TP_1$ . Contraction of $\beta$ and $\A_{P_1}$ gives rise to the following 1-form on $P_1$:
\[
\langle\beta, \A_{P_1}\rangle (p_1) = \langle\beta(p_1),\A_{P_1}(p_1)\rangle \in T^*_{p_1}P_1\,.
\]
If we denote the $TQ_2$-component of the transformation $\psi_{L_2,\beta}: T_{P_1}Q_1 \to T_{P_2}Q_2$ by $\psi_{L_2,\beta}^{TQ_2}$ (i.e. $\psi_{L_2,\beta}^{TQ_2}=\tau_1^{(2)}(\psi_{L_2,\beta})$), we have the following result:

\begin{theorem}\label{thm:main} Let $(F,f)$ be a transformation pair between $\e^{(1)}:P_1\to Q_1$ and $\e^{(2)}:P_2\to Q_2$ and let $(\e^{(2)},L_2,\B_2)$ be a magnetic Lagrangian systems such that $L_2$ is $f$-regular. Fix a connection $\A$ on the bundle $f:Q_2\to Q_1$ and a map $\beta:P_1\to V^*f$, and let $\psi_{L_2,\beta}: T_{P_1}Q_1 \to T_{P_2}Q_2$ be the $(F,f)$-compatible transformation constructed above. Consider the magnetic Lagrangian system $(\e^{(1)},L_1,\B_1)$ defined by
\begin{enumerate}[i)]
\item  $L_1(v_{q_1},p_1)=\left(\psi_{L_2,\beta}^*L_2\right)  (v_{q_1},p_1) -  \langle \beta
(p_1),\A(\psi_{L_2,\beta}^{TQ_2}(v_{q_1},p_1))\rangle$;
\item $\B_1= F^*\B_2 + d\left(\langle \beta , \A_{P_1}\rangle \right)$.
\end{enumerate}
Then $\psi_{L_2,\beta}$ satisfies:
\begin{enumerate}
\item $\psi_{L_2,\beta}^*\Omega^{L_2,\B_2} = \Omega^{L_1,\B_1}$;
\item $\psi_{L_2,\beta}^*E_{L_2} = E_{L_1}$.
\end{enumerate}
\end{theorem}
\begin{proof} The proof is straightforward generalization of the result in~\cite{BEC}.
\end{proof}

\paragraph{Example (Routh reduction).} The magnetic Lagrangian system on $\pi: Q\to Q/G$ has the following properties:
\begin{itemize}
\item $L_{1}= (i_{\mu}\circ \Pi_{\mu}^{-1})^{*}L - \Ac_{\mu}$, i.e. $L_{1}(T\pi(v_{q}) ,q)=  L(v_{q}) - \langle \mu, \Ac (v_{q})\rangle $ for $v_{q}\in J_{L}^{-1}(\mu)$ arbitrary,
\item $\B_{1}= d\Ac_{\mu}$.
\end{itemize}

\paragraph{The tangency condition.} Under some restrictive conditions (to be discussed below) the previous map can be proved to be a diffeomorphism. However, in general,  the magnetic Lagrangian system on $T_{P_1}Q_1$ is not regular even if the original system on $\e^{(2)}:P_2\to Q_2$ was. Theorem~\ref{thm:main} states that the two Lagrangian systems are related, but this does not guarantee that the solution curves to the (pre)-symplectic equations are related.

Here we use the results from Section~\ref{subset:pullback} to get a sufficient condition for solutions to the EL equations of $(\e^{(1)}:P_1\to Q_1, L_1, \B_1)$ to be related to those of a regular system $(\e^{(2)}:P_2\to Q_2, L_2, \B_2)$. From Section~\ref{subset:pullback} is it necessary and sufficient that $X_{E_{L_{2}}}$ is contained in the image of $T\psi_{L_{2},\beta}$.  The next proposition provides information on the image of $T\psi_{L_2,\beta}$.

\begin{proposition}\label{prop:tangency} Let $X_{s_{2}}$ denote an arbitrary tangent vector to $M=T_{P_2}Q_2$ at $s_{2}=(p_2,v_{q_2})=\psi_{L_2,\beta}(s_{1})$. Then $X_{s_{2}}=T\psi_{L_2,\beta}(Y_{s_{1}})$ for some $Y_{s_{1}}$ tangent to $N=T_{P_1}Q_1$ at $s_{1}=(p_1,v_{q_1})$  iff the following two conditions are satisfied:
\begin{enumerate}
\item $X_{s_{2}}$ and $Y_{s_{1}}$ are compatible;
\item $$X_{s_{2}}\left(\fpd{L_2}{{\bar  v}^a}\right) = \left(T\tau_2^{(1)}(Y_{s_{1}})\right)(\beta_a).$$
\end{enumerate}
\end{proposition}
\begin{proof} We first show that the two conditions hold if $X_{s_{2}}=T\psi_{L_2,\beta}(Y_{s_{1}})$. Since $\psi_{L_2,\beta}$ is a compatible map, the pair $X_{s_{2}},Y_{s_{1}}$ is compatible. Deriving the left hand side of the equality $\alpha_{L_2}\circ \psi_{L_2,\beta} = \beta \circ \tau_2^{(1)}$, becomes $T\alpha_{L_2}\circ T\psi_{L_2,\beta} (Y_{s_{1}}) = T\alpha_{L_2} (X_{s_{2}})$. The right hand side equals $T\beta \big(T\tau_2^{(1)}(Y_{s_{1}})\big)$. In components, we have $T(\alpha_{L_2})_a (X_{s_{2}}) = T\beta_a (T\tau_2^{(1)}Y_{s_{1}})$ which is the second condition.

For the converse statement, let $X_{s_{2}}, Y_{s_{1}}$ denote a pair of vectors satisfying 1. and 2. Note that the pair $T\psi_{L_2,\beta}(Y_{s_{1}}),Y_{s_{1}}$ also satisfies 1. and 2., and  that the proof is concluded if we can show uniqueness, i.e. two pairs $\bar X_{s_{2}},Y_{s_{1}}$ and $\bar X'_{s_{2}},Y_{s_{1}}$ satisfying conditions 1. and 2., will necessarily be equal: $\bar X_{s_{2}} = \bar X'_{s_{2}}$.

From the second condition (and using the coordinate expressions for compatible vectors given before in Equation~\eqref{eq:compavecx}) it follows
\[
(\bar X_{s_{2}} -\bar X'_{s_{2}})\left(\fpd{L_2}{{\bar  v}^a}\right)=0, \mbox{ or }
(\hat{\bar X}_{s_{2}}^b - \hat{\bar X}_{s_{2}}^{'b})\fpd{^2L_2}{{\bar  v}^a\partial {\bar v}^b} =0.
\]
Regularity of $\alpha_{L_2}$ implies uniqueness: $\bar X_{s_{2}}=\bar X'_{s_{2}}$.
\end{proof}

Denoting as before $\beta_a$ the component of $\beta$ along $d\bar q^a$, coordinate expressions for $\alpha_{L_2}$ and $\beta$ are:
\[
\alpha_{L_2}:(q^i,\bar q^a,v^i,\bar v^a,\bar p^\alpha)\mapsto\left(\fpd{L_2}{\bar v^a}\right),\hspace{.5cm}\beta:(q^i, \bar q^a,\bar p^\alpha, p^\gamma)\mapsto\beta_a.
\]
Taking tangent vectors points $s_{2}=\psi_{L_2,\beta}(s_{1})$, and using coordinate expressions for $X_{s_{2}}$ and $Y_{s_{1}}$ as in Equation~\eqref{eq:compavecx}, one finds that the equation ${(T\alpha_{L_2})}_b(X_{s_{2}}) = T\beta_b\big(T\tau_2^{(1)}Y_{s_{1}}\big)$ reads:
\begin{align}
Y_{s_{1}}^i\fpd{^2 L_2}{q^i\partial\bar v^b} + Y_{s_{1}}^a\fpd{^2 L_2}{\bar q^a\partial\bar v^b} &+ Y_{s_{1}}^\alpha\fpd{^2 L_2}{\bar p^\alpha\partial\bar v^b} +\hat Y_{s_{1}}^i\fpd{^2 L_2}{v^i\partial\bar v^b} + \hat X_{s_{1}}^a\fpd{^2 L_2}{\bar v^a\partial\bar v^b}\notag \\
&= Y_{s_{1}}^i\fpd{\beta_b}{q^i} + Y_{s_{1}}^a \fpd{\beta_b}{\bar q^a} + Y_{s_{1}}^\alpha \fpd{\beta_b}{\bar p^\alpha} + Y^\gamma_{s_{1}} \fpd{\beta_b}{p^\gamma}.\label{eq:master}
\end{align}

\paragraph{Example (Routh reduction).}  The tangency condition holds if $\beta$ is defined by a constant chosen momentum $\mu$. Given any vector $X_{s_{2}=v_{q}}$, then a compatible vector $Y_{s_{1} =(v_{[q]_{G}},q)}$ in $T(T_{Q}(Q/G))$ is completely determined from $X_{v_{q}}$. Equation~\ref{eq:master} can be rewritten as:
\[
X_{v_{q}}\left(\fpd{L_2}{{\bar  v}^a}\right) = \left(T\tau_{Q}(X_{v_{q}})\right)(\beta_a), \mbox{ with } \tau_{Q}:TQ\to Q.
\]
If $X_{v_{q}}=X_{E_{L}}(v_{q})$ the previous equation will be satisfied if $\beta$ is defined from a chosen fixed momentum $\mu$. The previous equation then becomes $X_{E_{L}}(J_{L}-\mu)=0$. It is well-known that this is satisfied: an invariant Hamiltonian vector field is tangent to the level set of a momentum map.

\paragraph{The diffeomorphic case.} The most interesting case of a compatible transformation $\psi_{L,\beta}$ arises precisely when this map is a diffeomorphism. In this case one has an induced system which is symplectomorphic to the original one, and hence its dynamics faithfully represent that of the original system. We prove here a useful condition for this to happen.

Assume that the dimensions of $T_ {P_i}Q_i$ agree, i.e. with the notations of Section~\ref{sec:transf} the following equality holds:  $2n_1+k_1=2n_2+k_2$. Then we have $n_1+k_1-n_2-k_2=n_2-n_1$, i.e. the dimensions of the fibers of $F$ and $f$ coincide, a necessary condition for $\psi_{L,\beta}$ to be a diffeomorphism. Note that in this case both the indices $a$ and $\gamma$ run from $1$ to $n_2-n_1$.

\begin{proposition}\label{pro:diff} In the situation above, assume the following regularity condition holds: the map $\beta_{|F^{-1}(p_2)}: F^{-1}(p_2) \to V^*_{q_2}f$ is a diffeomorphism for each $p_2 \in P_2$, with $\e^{(2)}(p_2) = q_2$. Then $\psi_{L,\beta}$ is a diffeomorphism.
\end{proposition}
\begin{proof} The fiber submanifold $F^{-1}(p_2)$ has coordinates $p^\gamma$, and in particular
\[
\rank (\partial\beta_a/\partial p^\gamma)_{a,\gamma}=n_2-n_1.
\]
The rank of $\psi_{L,\beta}$ is maximal iff $\rank (\partial\psi^a/\partial p^\gamma)_{a,\gamma}$ is maximal, where $\psi^a$ are the components of $\psi$, implicitly defined as:
\[
\fpd{L}{{\bar  v}^b}\left(q,\bar q,\dot q,\psi^a(q,\dot q,\bar q,\bar p, p),\bar p\right) = \beta_b\big(q,\bar q,\bar p, p\big).
\]
By $f$-regularity of $L$ it follows $\rank (\partial\psi^a/\partial p^\gamma)_{a,\gamma}=\rank (\partial\beta_a/\partial p^\gamma)_{a,\gamma}$. Since $\psi_{L,\beta}$ is a bijection (this is easily checked using the condition on $\beta$) and has constant maximal rank, the result holds.
\end{proof}
Moreover, from the proof it is clear that the previous proposition fully characterizes the case where $\psi_{L,\beta}$ is a diffeomorphisms, i.e., the condition on $\beta$ in Proposition~\ref{pro:diff} is also necessary. The following theorem guarantees the regularity of the induced systems under the transformation $\psi_{L,\beta}$ in this situation.
\begin{theorem} Assume $\psi_{L,\beta}$ is a diffeomorphism. Then the induced magnetic Lagrangian system on $\e^{(1)}:P_1\to Q_1$ is hyperregular.
\end{theorem}
\begin{proof} It is clear that $\F L_1$ is a global diffeomorphism, because  $\psi_{L,\beta}$ is a diffeomorphism. On the other hand, since $\psi_{L,\beta}$ is a symplectomorphism, it follows that the form $\big(\pi^{(1)}_1\big)^*\omega_{Q_1} +\big(\pi^{(1)}_2\big)^*\B_1$ is symplectic.
\end{proof}

\section{Fiberwise Reducible magnetic Lagrangian systems}\label{sec:transrouth}
This section is devoted to a particular kind of magnetic Lagrangian systems where the dynamics is easily reducible. These systems posses symmetry along the fibers of $\e:P\to Q$ which, roughly speaking, allows for a reduction of the base space space $P$ of $T_PQ$ while leaving the tangent part $TQ$ invariant.

\paragraph{Example (Routh reduction).} In Section~\ref{sec:transf} we have shown that a general Lagrangian system, amenable to Routh reduction, can be transformed into a magnetic Lagrangian system $(Q\to Q/G, (i_{\mu}\circ\Pi_{\mu}^{-1})^{*}L-\Ac_{\mu},d\Ac_{\mu})$ in such a way that the solution to the EL equations are mapped into solution of the original EL equations with fixed momentum $\mu$. There is still symmetry left in the transformed system: the fibers of $Q\to Q/G$ are equal to $G$ and we will show below that the transformed magnetic Lagrangian system is reducible under the fiberwise action of the isotropy subgroup $G_{\mu}$.

\paragraph{Preliminary results and definitions.}  Let $\e:P\to Q$ be a bundle and $\Phi^P$ denotes a $G$-action on $P$ such that  $\e\circ \Phi^P=\e$. Then $\Phi^{P}$ naturally induces a lifted action on $T_PQ$:
\[
\Phi^{T_PQ}_g(v_q,p)= (v_q,\Phi^P_g(p))=(v_q,gp),\;(v_q,p)\in T_PQ,\; g\in G.
\]

\begin{definition} A magnetic Lagrangian system $(\e:P\to Q,L,\B)$ together with a $G$-action $\Phi^P$ on $P$ is \emph{fiberwise-reducible} if the following conditions hold:
\begin{enumerate} \item The action of $G$ on $P$ is tangent to the fibers, i.e. $\e (\Phi^P_g (p))=\e (p)$;
\item $L$ is $G$-invariant with respect to the lift of $\Phi^P$ to $T_PQ$: $L(v_q,\Phi^P_g (p))=L(v_q,p)$;
\item The 2-form $\B$ on $P$ is reducible to $P/G$, i.e. $\B$ is $G$-invariant and satisfies $\imath_{\xi_Q} \B =0$ for all $\xi\in \lag$. \end{enumerate}
\end{definition}

We write $\bar{B}$ for the projection of $B$ onto $P/G$ and $\bar{L}$ for the projection of $L$ onto $T_{P/G} Q$.  The quotient manifold $P/G$ can be naturally fibered over $Q$, the fibration given by $\bar\e:[p]\mapsto \e(p)$. Since $\Phi$ is assumed to be free, the fibration $\bar\e:P/G\to Q$ is a principal $G$-bundle (that is, local triviality holds; see~\cite{sharpediff}). Since $\B$ projects to $\bar \B$, $\bar{\B}$ is closed.

\begin{definition} Let $(\e:P\to Q,L,\B,G)$ be a fiberwise-reducible magnetic Lagrangian system. We call $(\bar{\e}:P/G\to Q,\bar{L},\bar{\B})$ the \emph{(associated) reduced magnetic Lagrangian system}.
\end{definition}

We use the following notations, in agreement with the notations used before:
\begin{enumerate}
\item $\tau_G:T_PQ \to T_{P/G}Q$ is the projection that maps $(v_q,p )\in T_PQ$ onto $(v_q,[p] )\in T_{P/G}Q$.
\item $p_G:T_P^*Q \to T_{P/G}^*Q$ is the projection that maps $(\alpha_q,p )\in T_P^*Q$ onto $(\alpha_q,[p] )\in T_{P/G}^*Q$.
\item $\bar{\pi}_1:T_{P/G}^*Q \to T^*Q$ is the projection that maps $(\alpha_q,[p] )\in T_{P/G}^*Q$ onto $\alpha_q\in T^*Q$.
\item $\bar{\pi}_2:T_{P/G}^*Q \to P/G$ is the projection that maps $(\alpha_q,[p] )\in T_P^*Q$ onto $[p]\in P/G$.
\end{enumerate}

We are interested in reducing the dynamics in a fiberwise-reducible Lagrangian system to the associated reduced magnetic Lagrangian system $(\bar{\e}:P/G\to Q,\bar{L},\bar{\B})$. For that purpose, we need the following two lemmas:

\begin{lemma}\label{aux1} Let $(\e:P\to Q,L,\B,G)$  be a fiberwise-reducible Lagrangian system and consider the reduced magnetic Lagrangian system $(\bar{\e}:P/G\to Q,\bar{L},\bar{\B})$. Then $p_G\circ \F L=\F \bar{L}\circ \tau_G$, i.e., the diagram in Figure~\ref{fig:lemaux1} commutes.

\begin{figure}[htb]
\begin{center}
\includegraphics{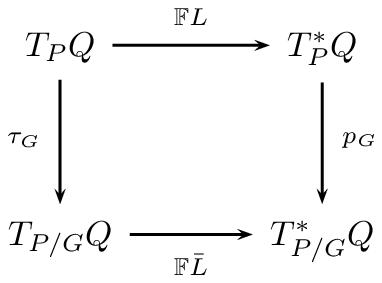}
\end{center}
\caption{Lemma~\ref{aux1}}\label{fig:lemaux1}
\end{figure}
\end{lemma}
\begin{proof} This result is immediate.
\end{proof}

\begin{lemma} The map $\tau_G$ between the presymplectic manifolds $(T_P Q,\Omega^{L,\B})$ and $(T_{P/G}Q,\bar{\Omega}^{\bar{L},\bar{\B}})$ satisfies $\tau_G^*\bar{\Omega}^{\bar{L},\bar{\B}}=\Omega^{L,\B}$ and $\tau_G^*E_{\bar{L}}=E_L$.
\end{lemma}
\begin{proof} The first statement follows from diagram chasing:
$$\tau_G^*\F \bar{L}^*(\bar{\pi}^*_1\omega_{Q} +\bar{\pi}^*_2\bar{\B})=\F L^*p_G^* (\bar{\pi}^*_1\omega_{Q} +\bar{\pi}^*_2\bar{\B})=\F L^* (\pi^*_1\omega_{Q} +\pi^*_2\B).$$
The second part is easily checked in coordinates:
$$\tau_G^*E_{\bar{L}}=\tau_G^*\left(\frac{\partial \bar{L}}{\partial v_i}v_i-\bar{L} \right)=\frac{\partial L}{\partial v_i}v_i-L=E_L.$$
\end{proof}

The tangency condition is trivially satisfied since the map $T_{P}Q \to T_{P/G}Q$ is a submersion. The next proposition summarizes the reduction of these fiberwise reducible systems.
\begin{proposition} Consider a fiberwise reducible magnetic system $(\e:P\to Q,L,\B,G)$.  The associated reduced magnetic Lagrangian system $(\bar{\e}:P/G\to Q,\bar{L},\bar{B})$ is such that any solution to the EL equation is the projection of a solution to the EL equations of the reducible system.
\end{proposition}

\paragraph{Example (Routh reduction).}  Clearly the Lagrangian $(i_{\mu}\circ\Pi_{\mu}^{-1})^{*}L-\Ac_{\mu}$ on $T_{Q}(Q/G)$ and the magnetic force term $d\Ac_{\mu}$ on $Q$ of the transformed Lagrangian system are $G_{\mu}$ invariant and fiberwise reducible to a magnetic Lagrangian system on $Q/G_{\mu}\to Q$. The reduced Lagrangian and magnetic 2-form correspond to  $\Ro^{\mu}$ and $\mathfrak{B}_\mu$ from Proposition~\ref{prop:normalrouth}.\\
\\
Throughout the previous sections, we have used the specific case of a standard Lagrangian system amenable to Routh reduction to demonstrate and develop the general theory on transformations between magnetic Lagrangian systems. In this section we conclude the main result of this paper: we summarize our statement that Routh reduction itself can be cast into the framework of the compatible transformations and we also consider the slightly more general framework of Routh reduction for magnetic Lagrangian systems (see~\cite{routhstages}). In both cases, Routh's reduction procedure for a `$G$-invariant'  magnetic Lagrangian system $(\e : P\to Q, L, \B)$ is realized as the result of two steps:
\begin{enumerate}[{\sc Step 1:}]
\item We construct an equivalent magnetic Lagrangian $(\e_\mu:P\to Q/G,L_{\mu},\B_\mu)$ system by means of a compatible transformation $\psi_{L,\beta}$ for a suitable $\beta$;
\item We check that $(\e_\mu:P\to Q/G,L_{\mu},\B_\mu)$ is fiberwise reducible.
\end{enumerate}
For a standard Lagrangian system $(Q\to Q,L,\B=0)$ amenable to Routh reduction, the procedure was demonstrated throughout the previous section.

\paragraph{Reduction for invariant magnetic Lagrangian systems.} Let $G$ act on $P\to Q$ by bundle automorphisms, i.e. there's a $G$-action on both $P$ and $Q$ such that $\e\circ\Phi^P_g=\Phi^Q_g\circ\e$. The projections of the principal bundles are denoted by $\pi^P:P\to P/G$ and $\pi^Q:Q\to Q/G$. This action naturally lifts to an action $\Phi^{T_PQ}$ on $T_PQ$ in the following way:
\[
\Phi^{T_PQ}_g(v_q,p):= (T\Phi^Q_g(v_q),\Phi^P_g(p)).
\]

\begin{definition} A magnetic Lagrangian system  $(\e:P\to Q,L,\B)$ is  $G$-invariant if $\B$ is invariant w.r.t. $\Phi^P$ and $L$ is invariant w.r.t. $\Phi^{T_PQ}$.\end{definition}
In this case, $\Phi^{T_PQ}$ is symplectic w.r.t. $\Omega^{L,\B}$ (recall the notations from Section~\ref{sec:background}). In order to obtain a momentum map for this action, we introduce the notion of $\B\lag$-potential.
\begin{definition}
Given an invariant closed 2-form $\B$ on $P$. Then a $\lag^*$-valued function $\delta$ on $P$ is a {\em $\B\lag$-potential} if $i_{\xi_{P}}\B = d\langle \delta, \xi \rangle$  for any $\xi \in \lag$.
\end{definition}
From now on and to ease notation, given a $\lag^*$-valued function $f$, $f_\xi$ for any $\xi\in\lag$ will be a shortcut for $\langle f,\xi\rangle$. For instance, the defining property of  a $\B\lag$-potential $\delta\in\mathcal{C}^\infty(P,\lag^*)$ is $i_{\xi_{P}}\B = d\delta_{\xi}$ for any $\xi\in\lag$. If $P$ is connected, we have
\begin{align*}
d\big[(\Phi^P_g)^* \delta_\xi\big] &= (\Phi^P_g)^* d\delta_\xi = (\Phi^P_g)^*(i_{\xi_{P}}\B)
=i_{(\Phi^P_g)^*\xi_{P}}(\Phi^P_g)^*\B\\
&=i_{(\Phi^P_g)^*\xi_{P}}\B=i_{(Ad_{g^{-1}}\xi)_P}\B=d\delta_{(Ad_{g^{-1}}\xi)}.
\end{align*}
From $d\big((\Phi^P_g)^* \delta_\xi - \delta_{Ad_g\xi}\big)=0$, it follows that the map $\sigma_\delta(g)=\delta\circ \Phi^{P}_g - Ad^{*}_{g^{-1}}\cdot \delta$ is a $\lag^*$-valued 1-cocycle on $G$. (This definition is independent of the point $p\in P$ because of connectedness). A momentum map $J_{L,\delta}$ for $\Phi^{T_PQ}$ is given by:
\[
\langle J_{L,\delta}(v_q,p), \xi\rangle  = \langle \F L(v_q,p), (\xi_Q(q),p)\rangle - \delta_\xi(p).
\]
This momentum map has non-equivariant cocycle $-\sigma_\delta$. We consider the affine action of $G$ on $\lag^*$ that makes $J_{L,\delta}$ equivariant (see for instance~\cite{Marsden}), and let $G_\mu$ denote the isotropy group of an element $\mu\in\lag^*$ w.r.t. this action. We will now prove that, under some regularity conditions, the level set of this momentum map may be identified with the subbundle $T_P(Q/G)\subset T_PQ$, and this identification will eventually allow us to define a suitable transformation scheme to describe Routh reduction on $T_PQ$.
\begin{definition} The Lagrangian $L$ of a $G$-invariant magnetic Lagrangian system is called $G$-regular if the map ${J_{L,\delta}}_{|(v_q,p)}: \lag \to \lag^*; \xi \mapsto J_{L,\delta}(v_q +\xi_Q(q),p) $ is a diffeomorphism for all $(v_q,p)\in T_PQ$.
\end{definition}
Similar to the standard case, we consider the map
\begin{align*}
\Pi_{\delta,\mu}:J^{-1}_{L,\delta}(\mu)&\to T_P(Q/G);\ (v_q,p) \mapsto (T\pi^Q(v_q),p).
\end{align*}
\begin{lemma}\label{prop:gregular} $\Pi_{\delta,\mu}$ is a diffeomorphism if the Lagrangian is $G$-regular.
\end{lemma}
\begin{proof} We define an inverse map $\Delta_{\delta,\mu}$ for $\Pi_{\delta,\mu}$. Choose an element $(v_{[q]_G},p)$ in $T_P(Q/G)$, and take a point $(v_q,p)\in T_PQ$ such that $T\pi^Q(v_q)=v_{[q]_G}$. Using $G$-regularity of $L$, there exists a unique $\xi$ in $\lag$ such that $J_{L,\delta}(v_q +\xi_Q(q),p) = \mu$. Set $\Delta_{\delta,\mu}(v_{[q]_G},p)=(v_q +\xi_Q(q),p)$. It remains to check that the construction is independent of the chosen point $(v_q,p)\in T_PQ$, but this is immediate.
\end{proof}
Consider now the $G_\mu$-action (on the second factor) on $T_P(Q/G)$, denoted by $\Psi^P_g$, i.e. $\Psi^P_g(v_{[q]_G},p)=(v_{[q]_G},\Phi^P_g(p))$ (this makes sense since $P/G_\mu$ fibers over $Q/G$). We know that $G_\mu$ acts on $J^{-1}_{L,\delta}(\mu)$, and it follows, for $g\in G_\mu$,
\begin{align*}
\Pi_{\delta,\mu}(T\Phi^{T_PQ}_g\cdot(v_q,p))&=\Pi_{\delta,\mu}(T\Phi^{TQ}_g(v_q),\Phi^P_g(p))=(T\pi^Q(T\Phi^{TQ}_g(v_q)),\Phi^P_g(p))\\
&=(T(\pi^Q\circ\Phi^Q_g)(v_q),\Phi^P_g(p))=(v_{[q]_G},\Phi^P_g(p))\\
&=\Psi^P_g\cdot\Pi_{\delta,\mu}(v_q,p),
\end{align*}
i.e. $\Pi_{\delta,\mu}(g\cdot(v_q,p))=g\cdot\Pi_{\delta,\mu}(v_q,p)$.  This implies equivariance for the inverse map $\Delta_{\delta,\mu}$.

\paragraph{The compatible transformation.} Analogous as the standard case, we consider the following transformation scheme: $P_{1}=P_{2}=P$, $Q_{1}=Q/G$, $Q_{2}=Q$ and transformation the pair $(F=id_P,f=\pi^Q=\pi : Q \to Q/G)$. We have $T_{P_{1}}Q_{1} = T_P(Q/G)$ and $T_{P_{2}}Q_{2}=T_PQ$, and $\pi$-regularity of $L$ is equivalent to $G$-regularity of $L$.

\begin{figure}[htb]
\begin{center}
\includegraphics{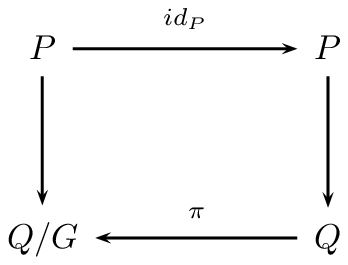}
\end{center}
\caption{Transformation scheme for $T_PQ$}
\end{figure}
We use the notations used in the compatible maps: we let coordinates on $Q/G$ be denoted $(q^{i})$, adapted coordinates on $Q$ are then $(q^{i},\bar q^{a})$ and finally $(q^{i},\bar q^{a},\bar p^\alpha)$ represent coordinates on $P$ (in particular, there are no components in $p^\gamma$).

The components of the infinitesimal generator of symmetries $\sigma :Q\times \lag \to TQ$ of $\Phi^Q$  are denoted by $\sigma^{a}_{b}$, i.e. $(q=(q^{i},\bar q^{a}),\xi=\xi_{b} e^{b})\mapsto \sigma^{a}_{b}(q^{i},\bar q^{a})\xi^{b}$, with $e_{b}$ a basis for $\lag$ (and $e^{b}$ denote the dual basis vectors). The action being free, $\sigma^{a}_{b}$ is invertible, $\Sigma^{a}_{b}:= (\sigma^{-1})^{a}_{b}$.

Define $\beta : P\to V^*\pi $ in the following way:
\[
\langle \beta(p), \xi_Q(q)\rangle  = \langle \mu, \xi\rangle + \langle \delta(p), \xi\rangle.
\]
In local coordinates, the map $\psi_{L,\beta}$ takes the form:
\[
\psi_{L,\beta} (q^{i},\bar q^{a},v^{i},\bar p^{\alpha}) = (q^{i},\bar q^{a},v^{i},\bar v^{a},\bar p^{\alpha}),
\]
with $\bar v^{\alpha}$ implicitly determined from
\[
\frac{\partial L}{\partial \bar v^{\alpha}} (q^{i},\bar q^{a},v^{i},\bar v^{a},\bar p^{\alpha}) = \beta_a=\Sigma^{b}_{a}\mu_b+\Sigma^{b}_{a}\delta_b,
\]
which is equivalent to the momentum equation and $\psi_{L,\beta}$ equals $\imath_\mu\circ\Pi_{\delta,\mu}^{-1}$. 

\paragraph{Verifying the tangency condition.}  Because the momentum map is conserved along solution to the EL equations, the tangency condition from~Proposition~\ref{prop:tangency} is fulfilled for any tangent vector that solves the (pre)-symplectic equation for the invariant Lagrangian system on $P\to Q$ at a point on the level set of the momentum map.   Here we will check Equation~\eqref{eq:master} for an arbitrary function $\beta$. More precisely, we study the tangency condition for a tangent vector $X_{s}$ to $T_{P}Q$ in the case of a compatible transformation map $\psi_{L,\beta}$ with $\beta$ arbitrary and where $X_{s=(v_{q},p)}$ solves the EL equation $i_{X_{s}} \Omega^{L,\B} = -dE_{L}$.

Because of the SODE nature EL equation, the tangent vector $X_{s}$ to $T_PQ$ is of the form:
\[
X_{s}= v^{i}\fpd{}{q^{i}} + \bar v^{a} \fpd{}{\bar q^{a}} + \dot p^{\alpha} \fpd{}{\bar p^\alpha} + \ddot q^{i}\fpd{}{v^{i}} + \ddot q^{a} \fpd{}{\bar v^{a}},
\]
where $(\ddot q, \dot p^{\alpha})$ are implicitly determined from the Euler Lagrange equations.
A tangent vector $Y_{\bar s=(T\pi(v_{q}),p)}$ to $T_{P}(Q/G)$ compatible to $X_{s}$ completely determined by this condition and is of the form
\[
Y= v^{i} \fpd{}{q^{i}}  + \bar v^a\fpd{}{\bar q^{a}} + \dot p^{\alpha} \fpd{}{\bar p^\alpha} + \ddot q^i \fpd{}{v^{i}}.
\]
From Proposition~\ref{prop:tangency}, $X_{s}$ is in the image of $T\psi_{L,\beta}$ if \[X_{s}\left(\fpd{L}{\bar v^{a}}\right) = Y_{\bar s} (\beta_{a}).\] Since $\beta$ is a function on $P$, the right hand side can be written (with a slight abuse of notation) as $X_{s}(\beta_{a})$, and the tangency condition becomes
\[X_{s}\left(\fpd{L}{\bar v^{a}}-\beta_{a}\right)=0.\]For a $G$-invariant Lagrangian, only $\beta_a=\Sigma^{b}_{a}\mu_b+\Sigma^{b}_{a}\delta_b$ will provide a transformation that satisfies the tangency conditions.

\paragraph{The reduction step.} Fix a principal connection $\mathcal{A}$ on the bundle $\pi: Q\rightarrow Q/G$ whose corresponding connection 1-form is denoted $\Ac$. We apply the construction of Theorem~\ref{thm:main} to induce the following magnetic Lagrangian system on $T_P(Q/G)$:
\begin{align*}
L_{\mu}&=(\imath_\mu\circ\Pi_{\delta,\mu}^{-1})^*L-\langle \mu+\delta,\Ac_P((\imath_\mu\circ\Pi_{\delta,\mu}^{-1})^{TQ})\rangle;\\
\B_\mu &= \B + d\langle \mu+\delta,\Ac_P \rangle.
\end{align*}
This new magnetic Lagrangian system $(P\to Q/G,L_\mu,\B_\mu)$ is $G_\mu$-fiberwise reducible because:
\begin{itemize}
\item$\Pi_{\delta,\mu}^{-1}$ is equivariant, $L$ is invariant, and the term involving $\Ac_P$ is $G_\mu$-invariant.
\item$\B_\mu$ is $G_\mu$ invariant and satisfies $\imath_{\xi_P} \B_\mu=0$ for all $\xi\in \lag_\mu$.
\end{itemize}
The last assertion can be checked using Cartan's formula and the fact that the infinitesimal 2-cocycle corresponding to $\sigma_\delta$ equals $\Sigma_\delta(\xi,\zeta)=-\langle T_e\sigma_\delta(\xi),\zeta\rangle=-\xi_P(\delta_\zeta(p))-\delta_{[\xi,\zeta]}(p)$. A detailed proof may be found in~\cite{routhstages,MarsdenHamRed}.

We conclude with a diagram (Figure~\ref{fig:Requiv}) that summarizes the equivalence of Routh reduction with the procedure described above: a transformation $\psi_{L,\beta}$ followed by a fiberwise reduction. The presymplectic structures on  $J_{L,\delta}^{-1}(\mu)$ and $T_P(Q/G)$ (the former given by $\imath_{\mu}^*\Omega^{L,\B}$ and latter given by Theorem~\ref{thm:main}) are related by $\Delta_{\mu,\delta}$. Finally, $\Delta_{\mu,\delta}$ drops to a symplectomorphism $\bar{\Delta}_{\mu,\delta}$ on the quotient.
\begin{figure}[htb]
\begin{center}
\includegraphics{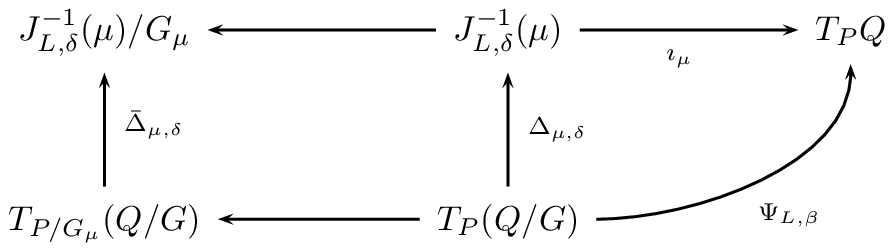}
\end{center}
\caption{Routh reduction scheme}\label{fig:Requiv}
\end{figure}

\section{The Hamiltonian picture}\label{sec:Ham} We end with a brief description of the Hamiltonian analogue of the transformations studied in the previous section. A full treatment of magnetic Hamiltonian systems as the Hamiltonian counterpart to magnetic Lagrangian systems is not the aim of this section, and will be addressed elsewhere. Accordingly we will only provide here the basic definition and properties.

\begin{definition} A magnetic Hamiltonian system is a triple $(\e:P\to Q, H, \B)$ where $\e:P\to Q$ is a fiber bundle, $H$ is a smooth function on the fiber product $T^*_PQ$ and $\B$ is a closed 2-form on $P$.
\end{definition}
There's a natural presymplectic form $\Omega$ defined on $T^*_PQ$, given by $\Omega:=\pi^*_1\omega_{Q} +\pi^*_2\B$. Given a magnetic Hamiltonian system, the dynamics associated to a magnetic Hamiltonian system are solutions to the Hamilton's equations w.r.t. the (pre)symplectic structure $\Omega$ and the Hamiltonian function $H$. This definition generalizes the standard definition of a Hamiltonian system on $T^*Q$ when considering $P=Q$.

\paragraph{Transformations} We will define a class of transformations $\psi_{\A,\beta}:T^*_{P_1}Q_1\to T^*_{P_2}Q_2$ analogous to the class $\psi_{L,\beta}$. Consider a transformation pair $(F,f)$ for the bundles $\e^{(1)}:P_1\to Q_1$ and $\e^{(2)}:P_2\to Q_2$ inducing adapted coordinates $(q^i,\bar q^a,\bar p^\alpha)$ on $P_2$ and $(q^i,\bar q^a,\bar p^\alpha,p^\gamma)$ on $P_1$ as in Section~\ref{sec:transf}. Corresponding coordinates in $T^*_{P_1}Q_1$ and $T^*_{P_2}Q_2$ are denoted $(q^i, \alpha^i,\bar q^a,\bar p^\alpha, p^\gamma)$ and $(q^i,\bar q^a,\alpha^i,\bar \alpha^a,\bar p^\alpha)$ respectively.

To determine the analogue of the transformation $\psi_{L_2,\beta}$ one begins with the following observation. The coordinate expression for the Lagrangian in Theorem~\ref{thm:main} induced by a transformation $\psi_{L_2,\beta}$ is
\[
L_1(q,\dot q,\bar q, \bar p,p) = \psi_{L_2,\beta}^*L_2(q,\bar q,\dot q,\dot {\bar q}, \bar p) - \beta_a(q,\bar q,\bar p,p) \big(\psi^a_{L_2,\beta}(q,\bar q,\dot q,\bar p,p) + \Gamma^a_i(q,\bar q)\dot q^i\big),
\]
where $\Gamma^a_i$ are the connection coefficients. A computation shows that the momenta $\alpha=\partial L_2/\partial \dot q$ transform under $\psi_{L_2,\beta}$ as $\alpha\mapsto \alpha + \langle\beta,\A\rangle$. More precisely, using the definition of $\psi_{L_2,\beta}$ one finds:
\begin{align*}
\psi_{L_2,\beta}^*\fpd{L_2}{{\dot{\bar q}}^a}=&\beta_a,\\
\psi_{L_2,\beta}^*\fpd{L_2}{\dot q^i}=&\fpd{L_1}{\dot q^i}-\big(\psi_{L_2,\beta}^*\fpd{L_2}{{\dot{\bar q}}^a}\big)\fpd{\psi^a}{\dot q^i}+\beta_a\big(\fpd{\psi^a}{\dot q^i}+\Gamma^a_i\big)=\fpd{L_1}{\dot q^i}-\beta_a\Gamma^a_i.
\end{align*}
Having the transformation law for the momenta, which depends on the chosen connection $\A$ and on the map $\beta$, one can naturally define a transformation $\psi_{\A,\beta}$ for a magnetic Hamiltonian systems on $P_2\to Q_2$ as the transformation which satisfies the aforementioned transformation law for the momenta (and covers $(F,f)$). The explicit expression is given by:
\[
\psi_{\beta,\A}(\alpha_{q_1},p_1)= \big(T_{q_2}^*f(\alpha_{q_1})+ \langle\beta,\A\rangle, F(p_1)\big),
\]
with $q_2=\epsilon^{(2)}(F(p_1))$. If one then defines the magnetic form $\B_1$ as
\[
\B_1= F^*\B_2 + d\left(\langle \beta , \A_{P_1}\rangle \right),
\]
one has the following result:
\begin{proposition} In the situation above, $\psi_{\A,\beta}^*\Omega_2=\Omega_1$.
\end{proposition}
\begin{proof} A point with coordinates $(q^i, \alpha^i,\bar q^a,\bar p^\alpha, p^\gamma)\in T_{P_1}Q_1$ is mapped into the point $(q^i,\bar q^a,\alpha_i+\beta_a\Gamma^a_i,\beta_a,\bar p^\alpha)\in T_{P_2}Q_2$ by $\psi_{\A,\beta}$. Using $\Omega_2=\big(d\alpha_i\wedge dq^i + d\bar\alpha_a\wedge d\bar q^a +\B_2\big)$, it follows easily
\begin{align*}
({\psi_{\A,\beta}})^*\Omega_2 &= d\alpha_i\wedge dq^i + d(\beta_a\Gamma^a_idq^i + \beta_ad\bar q^a)+F^*\B_2\\
&= d\alpha_i\wedge dq^i + \B_1=\Omega_1.
\end{align*}
\end{proof}
Starting from this result one defines the {\sl induced} magnetic Hamiltonian system on $T^*_{P_1}Q_1$, denoted $(\e^{(1)},H_1,\B_1)$, whose Hamiltonian function is given by
\[H_1(\alpha_{q_1},p_1)=\psi_{\A,\beta}^*H_2(\alpha_{q_2},p_2).
\]
As in the case of magnetic Lagrangian systems (see Theorem~\ref{thm:main}), one can then relate the dynamics of $(\e^{(2)},H_2,\B_2)$ to that of $(\e^{(1)},H_1,\B_1)$. We conclude this paragraph with an application of these transformations to a well known example from mechanics.

\vspace{\topsep}
\paragraph{Example (Momentum shift in cotangent bundle reduction).} We consider the standard setting for cotangent bundle reduction, namely a Hamiltonian system on $T^*Q$ and a $G$ action on $T^*Q$ by cotangent lifts with associated momentum map defined as:
\[
\langle J(\alpha_q),\xi \rangle = \langle \alpha_q,\xi_Q\rangle, \mbox{ for all } \xi\in\lag.
\]
The reduced space $(T^*Q)_\mu=J^{-1}(\mu)/G_\mu$ is usually realized by choosing an arbitrary principal connection $\Ac$ on $\pi:Q\to Q/G$ and making use of the so-called {\sl momentum shift}. The momentum shift map $S_\mu:J^{-1}(\mu)\to J^{-1}(0)$ is defined in terms of the connection as $S_\mu(\alpha_q)=\alpha_q-\Ac_\mu$. The target space $J^{-1}(0)$ is naturally identified with $T^*_Q(Q/G)$, and this allows for an easy reduction.
\begin{figure}[thb]
\begin{center}
\includegraphics{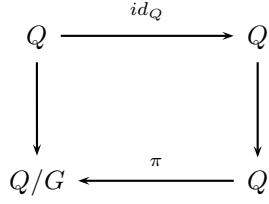}
\end{center}
\caption{Transformation scheme for $S_\mu$}\label{fig:diagramT^*Q}
\end{figure}

\noindent Consider the following scheme, which is the same as in Routh reduction: $Q_1=Q/G$, $Q_2=Q$, $P_1=Q$, $P_2=Q$ and the transformation pair $(F,f)=(id_Q,\pi)$. The situation is summarized in Figure~\ref{fig:diagramT^*Q}. The map $\beta$ is given by
\[
\langle \beta(q), \xi_Q(q)\rangle  = \langle \mu, \xi\rangle, \mbox{ for all } \xi\in\lag.
\]
It is easy to check that the map $\psi_{\A,\beta}:T^*_Q(Q/G)\to T^*Q,\;\alpha\mapsto\pi^*\alpha+\Ac_\mu$ equals $\imath_\mu\circ S_\mu^{-1}$ (where $\imath_\mu$ denotes the inclusion) and induces the magnetic Hamiltonian system on $T^*_Q(Q/G)$ whose Hamiltonian function and magnetic term are $(\imath_\mu\circ S_\mu^{-1})^*H$ and $d\Ac_\mu$ respectively. The situation is summarized in Figure~\ref{fig:MS}.
\begin{figure}[thb]
\begin{center}
\includegraphics{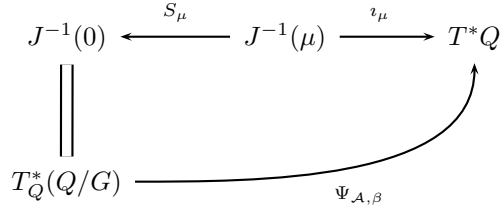}
\end{center}
\caption{Momentum shift}\label{fig:MS}
\end{figure}

\section*{Acknowledgments} BL is an honorary postdoctoral researcher at the Department of Mathematics of Ghent University. This work is sponsored by a Research Programme of the Research Foundation -- Flanders (FWO). This work is part of the {\sc irses} project {\sc
geomech} (nr.\ 246981) within the 7th European Community Framework Programme.

\end{document}